\documentclass[a4paper,UKenglish,cleveref,autoref,thm-restate]{article}
\usepackage{fullpage}
\usepackage[margin=1in]{geometry}

\usepackage{hyperref}
\hypersetup{colorlinks=true,citecolor=blue}
\usepackage{amsthm,amssymb,amsmath}

\usepackage{microtype}
\usepackage{xspace}
\usepackage{makecell}
\usepackage[ruled,noline,noend]{algorithm2e}
\usepackage{tabularx}
\usepackage{comment}
\usepackage{todonotes}
\usepackage{thmtools}
\usepackage{thm-restate}
\usepackage{authblk}
\usepackage[capitalise]{cleveref}
\usepackage{verbatim}
\usepackage{multicol}
\usepackage{enumitem}
\usepackage{cite}

\title{Sublinear Dynamic Interval Scheduling (on one or multiple machines)}

\author[1]{Pawe{\l} Gawrychowski}
\author[1]{Karol Pokorski}
\affil[1]{Institute of Computer Science, University of Wroc{\l}aw, Poland
[gawry,pokorski]@cs.uni.wroc.pl
}
\date{}

\newcommand{\cO}{\mathcal{O}}
\newcommand{\Oh}{\cO}
\newcommand{\cOtilde}{\tilde{\mathcal{O}}}
\newcommand{\Ohtilde}{\cOtilde}

\newcommand{\poly}{\operatorname{poly}}
\newcommand{\LC}{\textsc{LC}}
\newcommand{\LCinternal}{\textsc{LC-int}}
\newcommand{\LCdecr}{\textsc{LC-decr}}

\newcommand{\Rinternal}{\textsc{Res-int}}
\newcommand{\Einternal}{\textsc{Exit-int}}
\newcommand{\NextG}{\textsc{Next}}

\newcommand{\DS}{\mathcal{D}}
\newcommand{\PS}{\mathcal{P}}
\newcommand{\AS}{\mathcal{A}}
\newcommand{\BS}{\mathcal{B}}
\newcommand{\TS}{\mathcal{T}}
\newcommand{\TSm}{\overline{\mathcal{T}}}
\newcommand{\TT}{\mathfrak{T}}
\newcommand{\M}{\textsc{FMR}}
\newcommand{\polylog}{\operatorname{polylog}}

  \theoremstyle{plain}
  \newtheorem{theorem}{Theorem}
  \newtheorem{lemma}[theorem]{Lemma}
  \newtheorem{corollary}[theorem]{Corollary}
  \newtheorem{proposition}[theorem]{Proposition}

  \newtheorem{definition}[theorem]{Definition}
  
  \newtheorem{remark}[theorem]{Remark}
  
  \newtheorem{conjecture}[theorem]{Conjecture}

\newcommand{\FIGURE}[4]{
\begin{figure}[#1]
\begin{centering}
\includegraphics[width={#2}\textwidth]{figures/#3.pdf}
\caption{#4}
\label{fig:#3}
\end{centering}
\end{figure}
}

\begin{document}

\maketitle

\begin{abstract}
  We revisit the complexity of the classical Interval Scheduling in the dynamic setting.
  In this problem, the goal is to maintain a set of intervals under insertions and deletions
  and report the size of the maximum size subset of pairwise disjoint intervals after each
  update. Nontrivial approximation algorithms are known for this problem, for both
  the unweighted and weighted versions [Henzinger, Neumann, Wiese, SoCG 2020].
  Surprisingly, it was not known if the general exact version admits an exact solution working in
  sublinear time, that is, without recomputing the answer after each update.

  Our first contribution is a structure for Dynamic Interval Scheduling with amortized $\Ohtilde(n^{1/3})$ update time.
  Then, building on the ideas used for the case of one machine, we design
  a sublinear solution for any constant number of machines:
  we describe a structure for Dynamic Interval Scheduling on $m\geq 2$ machines
  with amortized $\Ohtilde(n^{1 - 1/m})$ update time.

  We complement the above results by considering Dynamic Weighted Interval
  Scheduling on one machine, that is maintaining (the weight of) the maximum
  weight subset of pairwise disjoint intervals. We show an almost linear
  lower bound (conditioned on the hardness of Minimum Weight $k$-Clique) for
  the update/query time of any structure for this problem. Hence, in the weighted case
  one should indeed seek approximate solutions.
\end{abstract}

\section{Introduction}

The \textsc{Interval Scheduling} (IS) problem is often used as one of the very first examples
of problems that can be solved with a greedy approach. In this problem,
we have a set of jobs, the $i$-th job represented by an interval $(s_{i},f_{i})$.
Given $n$ such intervals, we want to find a maximum size subset
of pairwise disjoint intervals. In this context, disjoint intervals are usually called
compatible. This admits a natural interpretation as a scheduling problem,
where each request corresponds to a job that  cannot be interrupted and require exclusive access to
a machine. Then, the goal is to schedule as many jobs as possible using a single
machine. The folklore greedy algorithm solves this problem in $\Oh(n)$ time,
assuming that the intervals are sorted by the values of $f_{i}$~\cite{kleinbergtardos}.
While it may appear to be just a puzzle, interval scheduling admits multiple applications
in areas such as logistics, telecommunication, manufacturing, or personnel scheduling.
For more applications and a detailed summary of different variants of interval
scheduling, we refer to~\cite{kolen}.

In many real-world applications, there is a need for maintaining the input under
certain updates (for example, insertions and deletions of items), so that we can
report the optimal solution (or its cost) after each operation. The goal is to
avoid the possibly very expensive recalculation of the answer (which surely
takes at least linear time in the size of the input) by maintaining some kind of
additional structure. The first step in this line of research is to design a structure
with sublinear update/query time. Then, the next goal is to bring down the
time complexities to polylogarithmic (in the size of the input). 
Examples of problems in which this has been successfully accomplished
include dynamic graph connectivity \cite{demainepatrascu,henzinger2,holm},
dynamic longest increasing subsequence \cite{kociumakaseddighin,gawrychowskijanczewski},
dynamic suffix array \cite{amir,kempakociumaka},
dynamic graph clustering \cite{doll}, and many others.
For some dynamic problems no such solutions are known,
and we have tools for proving (conditional) polynomial hardness
for dynamic algorithms~\cite{HenzingerKNS15}.

This suggests the following \textsc{Dynamic Interval Scheduling} (DIS) problem,
in which we want to maintain a set $S$ of intervals subject to insert and delete operations.
After each update, we should report the size of the maximum size subset of pairwise compatible
intervals. Note that reporting the subset itself might be not feasible, as it
might contain $\Omega(n)$ intervals. Similarly, neither is explicitly maintaining
this subset, as an update might trigger even $\Omega(n)$ changes in the unique
optimal subset. Thus, the challenge is to maintain an implicit representation
of the current solution that avoids recomputing the answer after each update, that is,
supports each update in sublinear time.
Besides being a natural extension of a very classical problem, we see this
question as possibly relevant in practical application in which we need to
cope with a dynamically changing set of jobs.

\subsection{Previous work}

Surprisingly, to the best of our knowledge, the complexity of general exact DIS
was not considered in the literature. However, Gavruskin et al.~\cite{gavruskin}
considered its restricted version, in which there is an extra constraint on the
set $S$. Namely, it should be \emph{monotonic} at all times: for any two intervals
$(s_i, f_i), (s_j, f_j) \in S$ we should have $s_{i}<s_{j}$ and $f_{i}<f_{j}$ or vice versa.
Under such assumption, there is a structure with $\Oh(\log^{2} n)$ amortized time per
update and $\Oh(\log n)$ amortized time per query. Alternatively,
the update time can be decreased to $\Oh(\log n)$ if the query only returns
if a given interval belongs to the optimal solution. 

For the general version of DIS, Henzinger, Neumann and Wiese~\cite{henzinger}
designed an efficient approximation algorithm that maintains an $(1+\epsilon)$-approximate
solution in polylogarithmic time. The dependency on $\epsilon$ has been very recently
improved from exponential to polynomial by Compton, Mitrović and Rubinfeld~\cite{Compton}.
In fact, both solutions work for the weighted version of the problem, called
\textsc{Dynamic Weighted Interval Scheduling} (DWIS). In this problem, each
interval has its associated weight, and the goal is to maintain a subset of pairwise
compatible intervals with the largest total weight. Note that the static
version of this problem, called \textsc{Weighted Interval Scheduling} (WIS),
can be solved by a straightforward dynamic programming algorithm~\cite{kleinbergtardos}
(but the greedy strategy no longer works now that we have weights).
This brings the challenge of determining if the unweighted (and weighted)
version of the problem admits an efficient exact solution.

A natural generalization of interval scheduling is to consider multiple
machines. In such a problem, there is a shared set of jobs to process, each job
can be either discarded or scheduled on one of the available $m$ machines. Jobs
scheduled on each machine must be pairwise compatible.
The goal is to maximize the number (or the total weight) of scheduled intervals.
IS on multiple machines (IS+) can be solved by extending the greedy algorithm
considering intervals by the earliest end time.
For each considered interval, if no machine is free at the respective time,
the interval is discarded. If there are some free machines, the interval is
assigned to the available machine that was busy at the latest. A direct
implementation of this approach incurs a factor of $m$ in the running time,
but this can be avoided~\cite{faigle,carlisle}.
The weighted version of the problem (WIS+) can be formulated and solved
as a min-cost flow problem~\cite{arkin,bouzina}.
For the dynamic version, Compton, Mitrović and Rubinfeld~\cite{Compton}
extend their methods for maintaining an approximate answer to multiple machines, however,
their bounds are mostly relevant for the unweighted case.
A related (but not directly connected) question is to maintain the smallest
number of machines necessary to schedule all jobs in the current set~\cite{gavruskin}.

\subsection{Our contribution}

In this paper, we consider dynamic interval scheduling on one and multiple machines.
We show that the unweighted version of the problem admits a sublinear dynamic
solution, and furthermore, we make non-trivial progress on decreasing the exponent
in the time complexity of the solution.

The starting point is a simple structure for the general DIS problem with
$\Oh(\sqrt{n} \log n)$ amortized update/query time. This is then improved
to $\Ohtilde(n^{1/3})$ amortized update/query time. For multiple machines, we
begin with $m=2$, and show how to solve the corresponding problem, denoted DIS2,
in $\Ohtilde(\sqrt{n})$ amortized time per update. Next, we use this solution
to solve the general DIS+ problem in $\Ohtilde(n^{1 - 1/m})$ amortized time per update.
While designing a solution working in $\Ohtilde(n^{1-1/(m+1)})$ time is not very difficult,
our improved time bounds require some structural insight that might be of independent
interest.

\begin{theorem}
  There is a date structure for Dynamic Interval Scheduling on $m \ge 1$
  machines that supports any update in  $\Ohtilde(\max(n^{1/3}, n^{1 - 1/m}))$ amortized time.
  \label{theorem:combined}
\end{theorem}

We complement the above result by a (conditional) lower bound for the weighted version
of the problem, even with $m=1$.  We show that, for every $\epsilon > 0$, under the  Minimum
Weight $(2\ell + 1)$-Clique Hypothesis, it is not possible to maintain
a structure that solves DWIS in $\Oh(n^{1-\epsilon})$ time per operation.
This shows an interesting difference between the static and dynamic complexities
of the unweighted and weighted versions: despite both IS and WIS admitting simple
efficient algorithms, DIS admits a sublinear solution while DWIS (probably) does not.

\subsection{Techniques and ideas}

A natural approach to DIS is to efficiently simulate the execution
of the greedy algorithm.

\begin{definition}
  For an interval $I_i = (s_i, f_i)$, the \emph{leftmost compatible} interval
  $\LC(I_i)$ is the interval $(s_{i'}, f_{i'}) \in S$ with the smallest
  $f_{i'}$ such that $s_{i'} \ge f_i$ and $\LC(I_i) = \perp$ if there is
  no such interval.
\end{definition}
Note that if the greedy algorithm includes $I_i$ in the solution then it also
includes $\LC(I_i)$. Thus, it is easy to prove that if $I_i$ is the
interval with the smallest $f_i$ in $S$, then
the (optimal) solution generated by the greedy algorithm is
$\{I_i, \LC(I_i), \LC^2(I_i), \dots\}$.

One can consider a forest in which each interval is represented as a node and
an interval $I_i$ has parent $\LC(I_i)$.
By creating an artificial root and connecting all forest roots' to it, we
make this representation a tree. We call it \emph{the greedy tree} (of $S$).
The answer to the DIS query is the length of the longest path from any node to
the root in the tree. We know this is actually the path from the earliest ending
interval thanks to the greedy algorithm.

\FIGURE{h}{0.63}{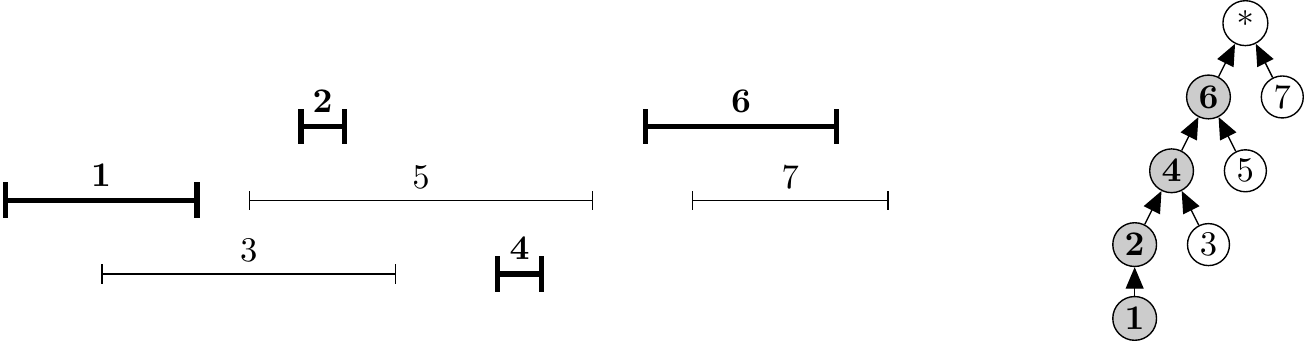}{
  An input instance for DIS with the optimal solution generated by the greedy
  algorithm marked using bold lines and the corresponding greedy tree.
}

A standard approach used in dynamic problems is splitting the current input
into several smaller pieces and recomputing some information only in the piece
containing the updated item. Then, the answer is obtained by using
the information precomputed for every piece. An attempt to use
such an approach for DIS could be as follows. We partition $S$ into parts,
either by the start or the end times, and in every part we precompute the result
of running the greedy algorithm from every possible state. The goal is to accelerate
running the algorithm by being able to jump over the parts.
For $m=1$, we can simply maintain the greedy tree, as it allows us to simulate
running the greedy algorithm not only from the interval with the smallest end time
but in fact from an arbitrary interval $I_{i}$. We call this
\emph{resuming the greedy algorithm from $I_i$}. This allows us
to jump over the whole part efficiently, and by appropriately balancing the size
of each part we obtain a data structure with $\Ohtilde(n^{1/2})$ time per update.
This is described in detail in \cref{section:n12_solution}.
A similar approach works for $m>1$, except that instead of the greedy tree
we need to preprocess the answer for every $m$-tuple of intervals,
resulting in $\Ohtilde(n^{1-1/(m+1)})$ time per update.

We improve on this basic idea for both $m=1$ and $m>1$. For $m=1$,
we design a way to solve the decremental variant of DIS in only (amortized)
polylogarithmic time per update, and couple this with maintaining a buffer
of the most recent insertions. 
For $m=2$, the greedy tree is no longer sufficient to capture all
possible states of the greedy algorithm. However, by a careful inspection,
we prove that for a piece consisting of $n$ intervals, instead of precomputing
the answers for all $\Theta(n^{2})$ possible states, it is enough to
consider only $\Oh(n)$ carefully selected states.
For $m>2$, we further extend this insight by identifying only $\Oh(n^{1-1/m})$
states, called \emph{compressible}. Interestingly, using these states
to simulate the greedy algorithm starting from an arbitrary state requires
a separate $\Oh(n^{2})$ precomputation, hence we need to consider the case $m=2$ separately.

\section{Interval scheduling on one machine}
\label{section:one_machine}

For our structures, it is sufficient that $s_i$ and $f_i$ characterizing
intervals are pairwise comparable but to simplify the presentation,
we assume that $s_i, f_i \in \mathbb{R_+}$.
One can also use an order maintenance structure \cite{dietz,bender} to
achieve worst-case constant time comparisons between endpoints even
if we only assume that when inserting an interval $(s_i, f_i)$ we know just the
endpoints of existing intervals in $S$ that are the nearest predecessors of
$s_i$ and $f_i$.
We make endpoints of all intervals pairwise distinct with the standard
perturbation.
We assume that each insert operation returns a handle to the interval
which can later be used to delete.

Our structures work in epochs. At the beginning of each epoch, we set $N$ to be
the number of intervals in $S$. When the number of intervals is outside range
$[\frac{N}{2}, 2N]$, the new epoch begins.
At the beginning of an epoch, we construct an additional data structure $\DS$ of
all intervals in $S$ by a sequence of inserts in any order.
These reconstructions have no impact on the amortized update time complexity
as $n$ actual operations are turned into $\Oh(n)$ insertions and deletions.
We maintain $\DS$ during the epoch.

We maintain a global successor structure storing all intervals sorted by their
end time that enables efficient computation of $\LC(\cdot)$.
There are $k$ separators that split the universe of coordinates into parts of
similar size.
Intervals are assigned into parts $\PS_0, \PS_1, \ldots, \PS_k$ by their start
time.
Some intervals are \emph{internal} (if they fully fit in the part)
and other are \emph{external} (otherwise).
Both $\Ohtilde(n^{1/2})$ and $\Ohtilde(n^{1/3})$ structures are able to
efficiently find the internal result for an interval $I_i$ in a part, that is how
many intervals the greedy algorithm can choose from $I_i$ until reaching the
exit (the last selected) interval of the part, so DIS query is solved by
iterating over these parts and applying the exit of one part as an input to the
next one.

We recommend reading \cref{section:n12_solution}, where we introduce the above
idea by showing a simpler but slower algorithm.
Here we extend this approach and present a data structure showing the
following.
\begin{theorem}
  There is a data structure for DIS that supports any sequence of $n$
  insert/delete/query on intervals in $\Ohtilde(n^{1/3})$ amortized time
  per each operation.
  \label{theorem:n13}
\end{theorem}

The separators are chosen such that each part has size at most $2N^{2/3}$ and
for any two consecutive parts $\PS_j$ and $\PS_{j+1}$ at least one has size at
least $\frac{1}{2}N^{2/3}$. Thus, there are always $\Oh(n^{1/3})$ parts.
More details on how to maintain this partition are provided in
\cref{section:n12_solution}.

Since our goal is to achieve $\Ohtilde(n^{1/3})$ update time and parts are
larger, we cannot afford to recompute the whole part from scratch for
every update in it (as we did in \cref{section:n12_solution}).
Instead, we keep internal intervals of a part in two structures:
a decremental structure and a buffer.
External intervals are only kept in the global balanced binary search tree
containing all the intervals.
We first sketch the idea and describe the details in the following subsections.

The decremental structure of each part contains $\Oh(n^{2/3})$ intervals,
has no information about buffer intervals, can be built in $\Ohtilde(n^{2/3})$
time and allows deletions in $\Oh(\polylog n)$ time.
The buffer $\BS_j \subseteq \PS_j$ contains only at most $N^{1/3}$ last inserted
internal intervals in $\PS_j$.
Each operation in a part leads to the recomputation of information associated
with the buffer in $\Ohtilde(n^{1/3})$ time.
When $\BS_j$ overflows, we rebuild the decremental structure from scratch
using all internal intervals from the part and clear the buffer.
Such recomputation happens every $\Omega(n^{1/3})$ updates
inside a part. This way the update time of our solution can still be within the
claimed bound.

As the optimal solution may use intervals both from the decremental collection
and the buffer interchangeably, we need to combine information stored for
these sets.
For buffer intervals, we can afford to precompute the whole internal result and
the exit of the part being fully aware of the content of the decremental
collection.
However, we also need to ``notify'' intervals of the decremental collection
about potential better solutions that can be obtained by switching to buffer
intervals.
For this we store an additional structure of total size of $\Ohtilde(n^{1/3})$,
recomputed every update in a part, specifying for which intervals of the
decremental collection there exists an ``interesting'' buffer interval.

\subsection{Active and inactive intervals}

\label{subsection:active_intervals}

\begin{definition}
An interval $I_i = (s_i, f_i)$ in a collection $C$ of intervals is
\emph{active} if there is no other $(s_{i'}, f_{i'}) \in C$ such that
$s_i \le s_{i'} \le f_{i'} \le f_i$. Otherwise $I_i$ is \emph{inactive}.
\label{def:active_inactive}
\end{definition}

\begin{lemma}
For any set $S$ of intervals and an interval $I_i \in S$, the greedy algorithm
for IS resumed from $I_i$ chooses (after $I_i$) only active intervals from $S$.
\label{lemma:only_active}
\end{lemma}
\begin{proof}
Assume there are two intervals $I_1 = (s_1, f_1)$ and $I_2 = (s_2, f_2)$
such that $s_2 < s_1 < f_1 < f_2$.
$I_1$ is considered earlier by the greedy algorithm.
If it is scheduled, $I_2$ can no longer be scheduled as $I_1$ and $I_2$ are
overlapping.
If it is not, $I_2$ also can not be scheduled as the set of compatible
intervals with $I_2$ is the subset of the compatible intervals with $I_1$.
\end{proof}

A collection of only active intervals is monotonic by definition.
This provides a well defined, natural order on the active intervals in the
collection:
$(s_i, f_i) \prec (s_{i'}, f_{i'}) \Leftrightarrow s_i < s_{i'} \Leftrightarrow f_i < f_{i'}$.
Because of this additional structure, we focus on describing how to maintain
the subset of active intervals inside a collection and only look for the
solution of (D)IS in this subset.

The decremental structure $\DS_j$ in each part only allows rebuilding and
deletions.
We maintain set of active intervals $\AS_j \subseteq \DS_j$ in the decremental
collection. When an interval from $\AS_j$ is deleted, the set should report
new active intervals.
We stress that the decremental structure is not aware of any buffer intervals of
$\PS_j$ and in order to determine if a particular interval is active in the
decremental collection we do not take into account any buffer intervals.

\begin{lemma}
  There is a structure that allows maintaining the subset of active intervals in
  a delete-only or insert-only collection of size $n$ in $\Oh(\log n)$ amortized
  time per insertion/deletion and can be built in $\Oh(n \log n)$ time.
  \label{lemma:decremental_active}
\end{lemma}

\begin{proof}
Each interval $(s_i, f_i)$ is translated into a point $(s_i, -f_i)$ in a plane.
We say that point $(x, y)$ \emph{dominates} point $(x', y')$ if $x > x' \wedge y > y'$.
Point $(x', y')$ is then \emph{dominated by} $(x, y)$. We say that
a point is \emph{dominated} if there is a point that dominates it.
The interval is active in the collection if and only if the point representing
it is not dominated.
The set of non-dominated points forms a linear order: the larger
$x$-coordinate implies the smaller $y$-coordinate. We store the front of
non-dominated points in a predecessor/successor structure.
Additionally, we maintain a range search tree indexed by $x$ storing in each
node the points of the appropriate range of $x$-coordinates and what is the
point with the maximum $y$ among them.

We start by describing the insert-only structure.
When a point $(x, y)$ is inserted, we search for its predecessor
$(x_\ell, y_\ell)$ and its successor $(x_r, y_r)$ in the front of non-dominated
points.
This way we can either find if $(x, y)$ is dominated by $(x_r, y_r)$ or if it
dominates $(x_\ell, y_\ell)$.
We then update the front and the range search tree appropriately.

To build the delete-only structure, we insert points one by one in any order
as described above.
When a point $(x, y)$ is deleted, we search for its predecessor
$(x_\ell, y_\ell)$ and its successor $(x_r, y_r)$ in the front and find what are
the points in the range $(x_\ell, x_r)$ that become non-dominated, that is
what are new maximums of nodes in the range search tree after removal of
$(x, y)$ from appropriate nodes.
These new non-dominated points are added to the front and each interval from the
decremental structure is activated only at most once. Thus, the time charged to
each interval in the collection is bounded by $\Oh(\log n)$.
\end{proof}

\subsection{Decremental structure}

For $I_i \in \PS_j$, we define $\LCdecr(I_i)$ to be the next greedy choice in
$\AS_j$ after $I_i$.

\begin{proposition}
  The set of greedy predecessors of $I_i$
  ($\{I_{i'}\ :\ \LCdecr(I_{i'}) = I_i\}$) forms a continuous range of active
  intervals in $\AS_j$.
  \label{proposition:predecessors_range}
\end{proposition}
\begin{proof}
  For any active intervals $I_1$, $I_2$, we have
  $I_1 \prec I_2 \Rightarrow \LCdecr(I_1) \preceq \LCdecr(I_2)$, so if
  there are three active intervals $I_1 \prec I_3 \prec I_2$ such that
  $\LCdecr(I_1) = \LCdecr(I_2)$ then also $\LCdecr(I_3) = \LCdecr(I_1)$.
\end{proof}

Intervals of $\AS_j$ form a forest where a node representing an
interval $I_i$ is the parent of $I_{i'}$'s node when
$\LCdecr(I_{i'}) = I_i$.
As in the previous section, we add an auxiliary
interval to make this representation a tree, we denote it $\TS_j$ and call it
\emph{a greedy tree of the part $\PS_j$}.
Greedy predecessors of $I_i$ are the children of node $I_i$ in the greedy tree.
We stress that the greedy tree is built only for the intervals of the
decremental collection.

We internally represent the greedy tree as an augmented top tree $\TT_j$
\cite{alstrup}.
This allows maintaining underlying fully dynamic forest (updates are
insertions/deletions of edges and changes to node/edge weights).
Because deletions and activations of intervals in the decremental structure
may change values of $\LCdecr(\cdot)$ for many nodes, we slightly alter the
structure as described in \cref{subsection:top_tree}.
This is also one of the reasons why one cannot apply techniques described in
\cite{gavruskin} to solve even the decremental variant of DIS despite being able
to efficiently maintain the (monotonic) set of active intervals.

When $I_i \in \AS_j$ is to be deleted, its children
$C = c_1 \prec c_2 \prec \dots \prec c_\ell$ have to connect to
other nodes of the greedy tree. Let $I_{k+1} = \LCdecr(I_i)$ before deletion and
$I_1 \prec I_2 \prec \dots \prec I_k$ are the activated intervals after removing
$I_i$. Note that $I_k \prec I_{k+1}$.
Other nodes than the elements of $C$ do not change its parent.

We first observe that $I_1, I_2, \dots, I_{k+1}$ are the only possible parents
for nodes in $C$, remind the fact that
$I_{i'} \prec I_{i''} \Rightarrow \LC(I_{i'}) \preceq \LC(I_{i''})$ and use
\cref{proposition:predecessors_range} to see that some (possibly empty) prefix
of children sequence ($c_1, c_2, \dots, c_{r_1}$) has to be connected to $I_1$,
then the next range ($c_{r_1 + 1}, c_{r_1 + 2}, \dots, c_{r_2}$) has to be
connected to $I_2$ and so on until finally some suffix of children sequence
($c_{r_k + 1}, c_{r_k + 2}, \dots, c_\ell$) has to be connected to $I_{k+1}$.
We use binary search on the children sequence to find indices
$r_1, r_2, \ldots, r_k$ in this order.
We update the parents of the nodes in the found ranges in the greedy tree as
described in \cref{subsection:top_tree} and it takes $\Oh(\polylog n)$ per
each \textbf{activated} interval.

Using the appropriate query to the top tree, we can resume the execution of the
greedy algorithm restricted to $\AS_j$ from any $I_i \in \AS_j$ in
$\Oh(\polylog n)$ time.

\subsection{Top tree}

\label{subsection:top_tree}

The underlying information maintained in $\TT_j$ is chosen to compute the
following:
\begin{itemize}
  \setlength{\itemsep}{0pt}
  \item weighted level ancestors,
  \item nearest marked ancestors,
  \item the total path weight from a node to the root (the sum of weights).
\end{itemize}
The discussion on how to maintain information that allows efficient computation
of the above in $\TT_j$ can be found in \cite{alstrup}.

$\TT_j$ represents an underlying modified greedy tree $\TSm_j$, namely, we
binarize the tree by reorganizing the children of each non-leaf node and adding
auxiliary nodes as presented in \cref{fig:ladder}.
A node in such a modified greedy tree that represents an actual interval has
a weight $1$, all other auxiliary nodes have a weight $0$.
The weight of the path between nodes is the sum of the weights of the nodes on
the path (including the endpoints).
This way, the weight of a path from a node representing an interval $I_i$ to the
root of the modified greedy tree represents the number of intervals chosen by
the greedy algorithm from $I_i$.

\FIGURE{t}{0.44}{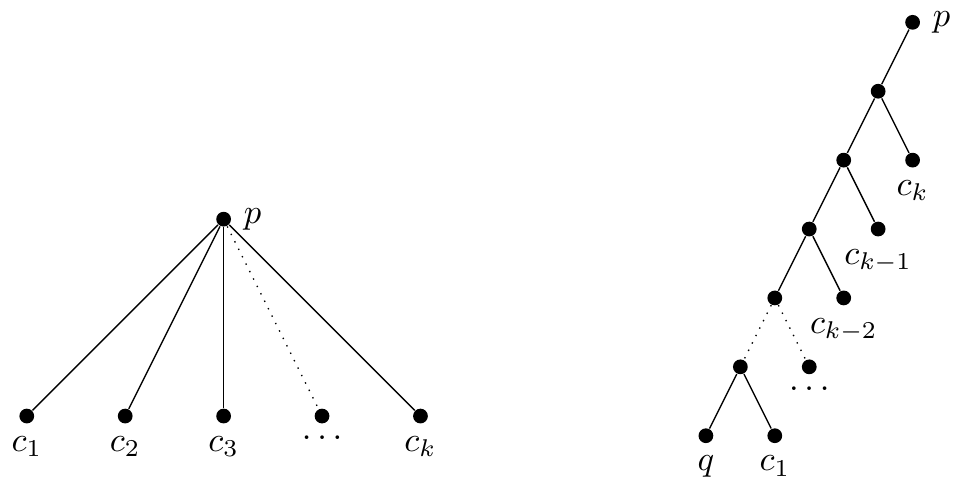}{
  A part of a greedy tree is shown on the left and the modified greedy tree
  represented by $\TT_j$ is shown on the right.
  We assume $c_1 \prec c_2 \prec c_3 \prec \ldots \prec c_k$.
  One can retrieve $i$-th child of a node by querying for the level ancestor
  from node $q$ (ignoring the weights of nodes).
}
$\AS_j$ is always monotonic so we use $\prec$ order on children.
This way, we can update values of $\LCdecr(\cdot)$ for
a range of children of a node in $\Oh(\polylog n)$ time by the appropriate
splits and joins in $\TT_j$. Apart from auxiliary nodes,
pre-order traversals of $\TS_j$ and $\TSm_j$ are equal.

$\TSm_j$ and $\TT_j$ are only internal representations of $\TS_j$
that enable efficient implementation of the necessary operations. Any updates of
$\TS_j$ are naturally translated into updates of $\TSm_j$ and $\TT_j$ or were
described above. We proceed with describing the further details on $\TS_j$.

\begin{definition}
  For an interval $I_i \in \AS_j$, we define its \emph{depth} as the depth in
  $\TS_j$. The set of intervals of the same depth $d$ is called \emph{a layer}
  $d$ in $\TS_j$ (or $\PS_j$).
\end{definition}

\begin{remark}
  We already have all the ingredients for the algorithm to solve the delete-only
  DIS variant in $\Oh(\polylog n)$ time.
  In this case, we do not partition intervals nor use a buffer.
  Instead, we only use the top tree representing the greedy tree of all the
  active intervals in the whole decremental collection of intervals.

  Similarly, we remind that the structure for maintaining the subset of active
  intervals can be also maintained for the insert-only variant of DIS
  (\cref{lemma:decremental_active}). Now we also observe that we can maintain
  the greedy tree when the intervals are only inserted.
  A new interval $I_i$ may only improve $\LC(\cdot)$ for some continuous range of
  intervals and we can binary search the endpoints of this range.
  To account for the cost of reconnecting these nodes, which may have many
  different parents, we observe that for any insertion, there is only at most
  one interval that loses a child in the greedy tree and is not deactivated.
  We charge the time of reconnection of the range of its children to the
  insertion of $I_i$.
  We charge the time needed to reconnect other nodes to the insertion of their
  (deactivated, thus actually deleted) parent.
  This establishes the time complexity of the insert-only variant of DIS to
  $\Oh(\polylog n)$.
\end{remark}

\subsection{Buffer}
\label{subsection:buffer}


\begin{definition}
  For intervals $I_1 \in \AS_j$ and $I_2 \in \BS_j$, we say that $I_1$
  \emph{directly wants to switch} to $I_2$ if and only if all the following
  conditions hold:
  \begin{itemize}
    \setlength{\itemsep}{0pt}
    \item $I_1$ ends earlier than $I_2$,
    \item $I_1$ and $I_2$ are compatible,
    \item $I_2 \prec \LCdecr(I_1)$.
  \end{itemize}
\end{definition}

The aim of the above definition is to capture that sometimes the value of
$\LC(\cdot)$ may be different from $\LCdecr(\cdot)$. Note that if
$I_1$ directly wants to switch to $I_2$ it does not necessarily imply that
$\LC(I_1) = I_2$. It just means that $I_2$ is (in sense of $\prec$)
a better next greedy choice for $I_1$ than it appears from the computation in
the decremental collection.
Note that it also means that the greedy algorithm resumed from any node in the
subtree of $I_1$ in the greedy tree will not choose $\LCdecr(I_1)$.
Thus we define the following.
\begin{definition}
  For intervals $I_1 \in \AS_j$, $I_2 \in \BS_j$ we say that $I_1$
  \emph{wants to switch} to $I_2$ if and only if there exists an integer
  $k \ge 0$ such that $\LCdecr^k(I_1)$ directly wants to switch to $I_2$.
\end{definition}

\FIGURE{h}{0.65}{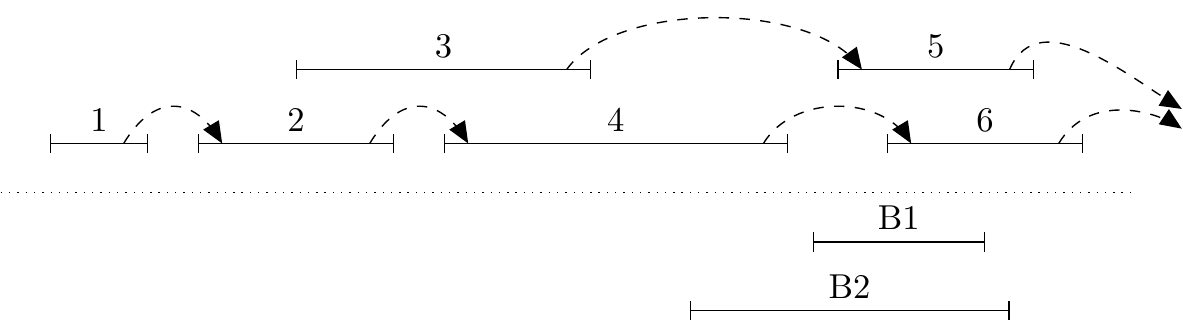}{
  An instance of DIS, example part. Intervals in the decremental collection are
  shown above the dotted line and buffer intervals are below. Dashed arrows
  connect intervals with their respective $\LCdecr(\cdot)$. Here intervals
  1, 2, 3 and 4 want to switch to B1 (3 and 4 directly) and interval 3 wants to
  (directly) switch to B2.
}

\begin{proposition}
  For an interval $I_i \in \BS_j$, there exist an integer $d$ such that the set
  of intervals in $\AS_j$ that directly want to switch to $I_i$ is either:
  \begin{itemize}
    \setlength{\itemsep}{0pt}
    \item a continuous range of a layer $d$,
    \item a suffix of layer $d$ and a prefix of layer $d+1$.
  \end{itemize}
  \label{prop:switching_ranges}
\end{proposition}
\begin{proof}
  Let $I_1 \prec I_3 \prec I_2$ and assume that $I_1$ and $I_2$ want to switch
  to $I_i$. Then, also $I_3$ wants to switch to $I_i$:
  $I_i$ ends earlier than $\LCdecr(I_3)$ because $I_1$ wants to switch
  and $I_3$ can switch to $I_i$ because $I_2$ can. This shows that the nodes
  that want to switch to $I_i$ form a continuous range in $\prec$.
  Active intervals that directly want to switch to any particular $I_i$ are
  pairwise overlapping. Indeed, with of any two compatible intervals
  $I_1 \prec I_2$, we would have $\LC(I_1) \preceq I_2$ so $\LC(I_1)$ ends
  earlier than any buffer interval compatible with $I_2$ to the right of $I_2$.
  This also proves that a node and its parent in the greedy tree cannot both
  directly want to switch to the same buffer interval thus completing the proof.
\end{proof}

\cref{prop:switching_ranges} shows that the actual size of the information
needed to notify the intervals from $\AS_j$ that want to directly switch to
a particular buffer interval is short. For each buffer interval, it is enough to
remember endpoints of at most two ranges.

We want to efficiently store also indirect switching.
Intervals that want to switch to $I_i$ are the nodes in subtree of any node in
ranges from \cref{prop:switching_ranges}.
For range from $I_1$ to $I_2$ on layer $d$ that wants to directly switch to
$I_i$, any node $I_3$ on layer $d' = d + k \ge d$ satisfying
$I_1 \preceq \LCdecr^k(I_3) \preceq I_2$ wants to switch to $I_i$, see
\cref{fig:switching}.

Note that if we traverse the greedy tree in BFS order (visiting children
left-to-right) we obtain exactly $\prec$ order.
Thus, when comparing two intervals on the same layer we can just see
which one is earlier in the pre-order traversal of $\TS_j$. This way we can
treat layers as sorted collections of intervals (actually, subranges of
$\prec$).

We use a 2D range search tree indexed by depth and position in the pre-order
traversal of intervals in $\AS_j$.
The structure allows us to store a collection of three-sided rectangles, so that
given query point we can check if it is contained in at least one of the
rectangles.
To mark nodes as in \cref{fig:switching} we add
$[d, +\infty) \times [\operatorname{pre}(I_1), \operatorname{pre}(I_2)]$
to the tree.

\FIGURE{t}{0.36}{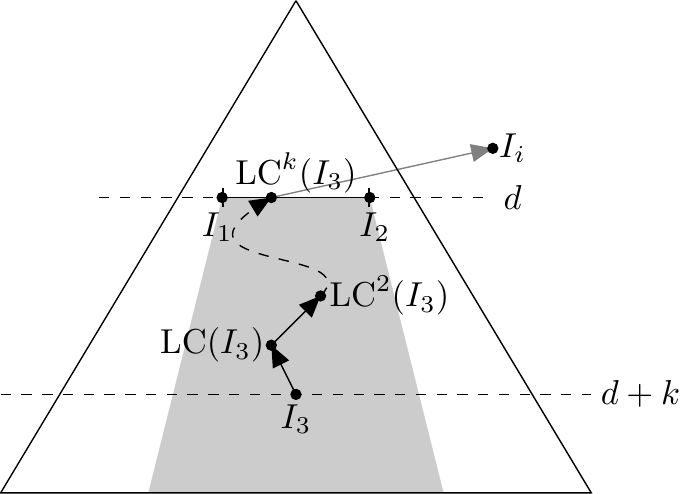}{
  Indirect possibility of switching to $I_i \in \BS_j$ for $I_3 \in \AS_j$.
  Nodes of $d$-th layer in range from $I_1$ to $I_2$ directly want to switch to
  $I_i$.
  In the gray area are the nodes that want to switch to $I_i$.
  \label{fig:switching}
}

An interval may want to switch to multiple intervals but the actual switching
point for any $I_i \in \AS_j$ is the earliest in $\prec$ (the deepest in
$\TS_j$) interval that wants to directly switch to a buffer interval on the path
from $I_i$ to the root in $\TS_j$.
We can deduce the actual earliest switching to buffer interval from any
$I_i \in \AS_j$ on layer $d$ in $\Oh(\polylog n)$ time by using a binary search
on depth $d' \le d$, each time querying the 2D range search tree if a point
$(d', I_i)$ is covered by at least one rectangle. The result for the prefix
of the path until reaching the buffer can be obtained from the top tree $\TT_j$.
We recreate the whole range search tree after an update in the part.

For any $I_i \in \BS_j$ we store the total length of the path to the root of
$\TS_j$ (this is the internal result for $I_i$ in $\PS_j$) and the latest actual
interval of $\PS_j$ just before reaching the root (this is the exit for $I_i$
in $\PS_j$).
This information is recomputed for all buffer intervals in $\PS_j$ using dynamic
programming by iterating the buffer intervals by decreasing end times as
follows.
For $I_i \in \BS_j$, we compute $\LC(I_i)$ and if it is a buffer interval, we
use its exit result and its internal result plus $1$ as the information for
$I_i$ (and, by the order of the computation, we already know these).
If $\LC(I_i) \in \AS_j$, we query the decremental collection for the next buffer
interval after $I_i$ selected by the greedy algorithm as described above and
combine its result with the prefix of the traversed path from $\LC(I_i)$ in the
decremental collection.
This is computed in $\Ohtilde(n^{1/3})$ time.

\section{Interval scheduling on multiple machines}

We stress that we assume that there are constant number of machines thus we are
going to ignore $\Oh(\poly m)$ factors in time complexities.
The difference between naive application of standard techniques and our
algorithms is negligible when $m$ is large.

As the main idea of our algorithm is to efficiently simulate the folklore greedy
algorithm for IS+ (described in \cite{faigle,carlisle}), we now remind it.
The intervals are considered separately by the earliest end time. For each
considered interval, if there is no available machine at the time, the job is
rejected. Otherwise, it is accepted and assigned the available machine that
was busy at the latest time. The proof of correctness is a standard exchange
argument.

The state of the partial execution (up to some time $t$) of the greedy algorithm
can be fully described by the sequence of length $m$, where $i$-th entry
describes which interval was last scheduled on $i$-th machine before or at time
$t$.
Some of the entry intervals to $\PS_j$ may not belong to $\PS_{j-1}$ if some
machine had not accepted any intervals in $\PS_{j-1}$.
At the same time, we want to preprocess information only for tuples of intervals
from $\PS_j$, thus we need the following additional notation.
\begin{definition}
The \emph{greedy state} $G_t$ (at time $t$) is the (multi)set of $m$ input
intervals.
Each element $I_i = (s_i, f_i) \in G_t$ means that at time $t$ there is
a machine that was busy up to time $f_i$. We use elements $\overline{I_i}$
to indicate that there is a machine which was busy up to time $s_i$.
\end{definition}
$\overline{~\cdot~}$ indicates that the particular machine is blocked for all
intervals that start too early.
Thus, despite each interval can only be selected once, we may want to mark that
some machines are busy up to the same time. For this reason, we decided to use
multisets for greedy states.
$\overline{(s_i, f_i)}$ can be simulated by an artificial interval
$(-\infty, s_i)$.

The greedy algorithm only considers values of $t$ that are end times of
intervals $I_i = (s_i, f_i)$ in the input.
We slightly abuse the notation and use $G_k$ to denote the greedy state at
time $f_k$ and assume the intervals are ordered according to the order of the
IS+ algorithm i.e.
$f_1 < f_2 < \ldots < f_n$.
To not consider cases with $|G_k| < m$ we add
$m$ pairwise overlapping intervals ending all earlier than the beginning
of any actual input interval.

If $G_{k-1} \neq G_k$, exactly one element of $G_{k-1}$ needs to be updated
to obtain $G_k$.
It is the one that is ending the latest among the elements of $G_{k-1}$
compatible with $I_k$.
One can see the same from a slightly different perspective.
Let assume that $i$ is the index for which
$\LC(I_i)$ is the earliest ending interval among
$G_k$. Then $G_k = G_{k+1} = \ldots = G_{k'-1} \neq G_{k'}$ and
$G_{k'} = G_k \setminus \{I_i\} \cup \{\LC(I_i)\}$. We call $G_{k'}$
the \emph{next greedy state after $G_k$} and denote it $\NextG(G_k)$.
Because $\LC(\cdot)$ can be computed in $\Oh(\polylog n)$ time using
the appropriate structure as described in \cref{section:n12_solution},
we iterate through all candidates for $I_i$ in the greedy state and thus have
the following.

\begin{corollary}
  $\NextG(G_k)$ can be computed in $\Ohtilde(m)$ time for any $G_k$.
  \label{prop:next_greedy_state}
\end{corollary}

We use insights from \cref{section:one_machine} and \cref{section:n12_solution}
and split the intervals into $\Oh(n^{1 - 1/m})$ parts of size at most
$\Oh(n^{1/m})$. But now the part
to which the interval belongs is determined by its end. Other
details like epochs, splitting and merging the parts remain the same.
We restrict $\LC(\cdot)$ to only consider intervals in the same part as the
argument of the operation (it can return $\perp$).
We build an additional structure for internal intervals in each part and rebuild
it every update in the part. As in the case of interval scheduling on one
machine, our goal is to be able to efficiently handle (in $\Oh(\polylog n)$
time) a query for the internal result (the number of accepted intervals) and the
exit greedy state from the part for a given entry greedy state $G_k$ in the part
-- we call this the \emph{part query from the greedy state $G_k$}.

Notice that during the execution of the greedy algorithm up to $\PS_{j-1}$,
it may happen that some machine will not accept any new interval in $\PS_{j-1}$,
so the exit greedy state coming from $\PS_{j-1}$ may contain intervals also from
earlier parts.
Let us now describe how to translate such an exit greedy state coming from
$\PS_{j-1}$ into an entry greedy state of $\PS_j$, so we can later only consider
the content of one part.
We observe that the decisions of the greedy algorithm only depend on the
relative order of endpoints of the considered intervals.
If a machine was busy up to time $t$ and there are no intervals starting before
time $t' > t$, we can safely assume that the machine is busy up to time $t'$
without changing the execution of the greedy algorithm.
Thus, we round up the end of each interval in the greedy state to the earliest start
of some interval in $\PS_j$. See \cref{fig:parts_merge}.
We stress that the result of rounding is not necessarily part of the solution
generated by our algorithm. It just indicates times up to which the machines
are busy.
After computing the exit greedy state for $\PS_j$, we inspect if there are
machines that have not accepted any intervals from $\PS_j$ and revert the
rounding for these.

\FIGURE{t}{0.36}{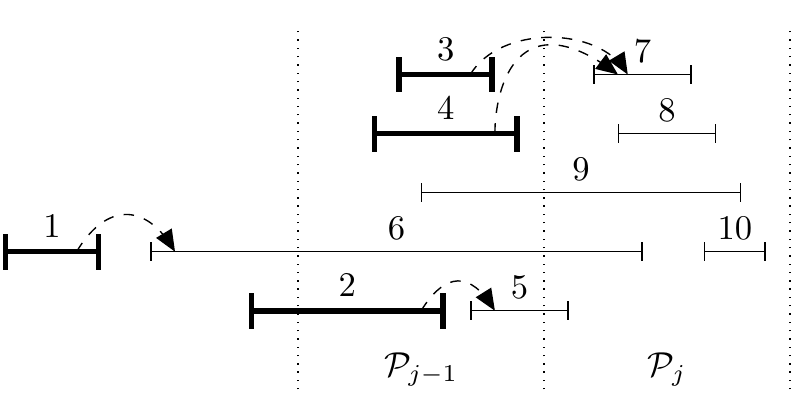}{
  Translation of an exit greedy state $G = \{1, 2, 3, 4\}$ from part
  $\PS_{j-1}$.
  Each interval of $G$ is rounded to the earliest starting interval in $\PS_j$
  that is later than the end of the interval (denoted by dashed directed edge).
  Thus, we can assume that the entry greedy state in $\PS_j$ is
  $\{\overline{5}, \overline{6}, \overline{7}, \overline{7}\}$.
  \label{fig:parts_merge}
}

We stick to \cref{def:active_inactive}, but we cannot make direct use
of \cref{lemma:only_active} because in the case of multiple machines it may
happen that inactive intervals are part of the optimal solution.
As these intervals may not form a monotonic collection, we redefine $\prec$
order as follows:
$(s_1, f_1) \prec (s_2, f_2) \Leftrightarrow f_1 < f_2$.
We still maintain the greedy tree $\TS_j$ and the top tree
$\TT_j$\footnote{We could also use simpler structures as we only need a
subset of operations provided by the top tree and we can afford to rebuild the
structure from scratch every update.} as described for one machine.
We identify intervals with the nodes representing them in $\TS_j$.

\begin{lemma}
  Let $I_i = (s_i, f_i)$ be an inactive interval and let
  $I_{i'} = (s_{i'}, f_{i'})$ be the latest (in $\prec$)
  interval contained inside $I_i$.
  If $I_i \in G_i$ then also $I_{i'} \in G_i$.
  \label{lemma:inactive_active_inside}
\end{lemma}
\begin{proof}
  Any interval compatible with $I_i$ is also compatible with $I_{i'}$ and
  $I_{i'}$ ends earlier than $I_i$.
  This means that if $I_i$ is accepted then also $I_{i'}$ is (at time $f_{i'}$).
  From time $f_{i'}$ up to time $f_i$ the machine that accepted
  $I_{i'}$ cannot accept other interval: it would have to start
  after $f_{i'}$ and end before $f_i$ thus violating our assumption that
  $I_{i'}$ is the latest interval contained inside $I_i$. This implies that
  $I_{i'} \in G_i$.
\end{proof}

\begin{lemma}
  Let $G = \{I_1, I_2, \ldots, I_m\}$ be a greedy state for which all elements
  are active intervals.
  Let $I_* = (s_*, f_*)$ be the earliest (in $\prec$) interval being a common
  ancestor of any pair of elements of $G$. Let $\operatorname{N}(I_i, I_*)$ be
  a prefix of $\{I_i, \LC(I_i), \LC^2(I_i), \ldots\}$ intervals preceding
  $I_*$ and let $\operatorname{S}(I_i, I_*)$ be the latest of
  $\operatorname{N}(I_i, I_*)$.

  Then
  $\operatorname{N}(I_1, I_*) \cup \operatorname{N}(I_2, I_*) \cup \ldots \cup \operatorname{N}(I_m, I_*)$
  are the only elements scheduled by the greedy algorithm for IS+ resumed from
  $G$ before reaching time $f_*$.
  Additionally, just before time $f_*$ the greedy state of the algorithm is
  $\{ \operatorname{S}(I_i, I_*) : i \in \{1, 2, \ldots, m\} \}$.
  \label{lemma:greedy_tree_skip}
\end{lemma}
\begin{proof}
  First, we make a technical note that thanks to the artificial root added to
  form the greedy tree, the interval $I_*$ always exists.

  The candidates for values of $\NextG(G_t)$ are
  $G_t$s with exactly one of the intervals replaced by its $\LC(\cdot)$
  assigned to the same machine.
  Thus, by using this reasoning inductively for $\NextG^k(G)$ for increasing
  $k$, we observe that when moving forward along the path from any $I_i \in G$
  to the root in the greedy tree, at least until reaching some
  interval $\succeq I_*$, all the traversed intervals will be scheduled on the
  same machine as $I_i$. Additionally, for different $I_{i'}, I_{i''} \in G$,
  the paths from $I_{i'}$ and $I_{i''}$ in the greedy tree do not share any
  nodes that are $\prec I_*$ (by definition).
  This way, all and the only elements that are included in some greedy state
  after $G$ before considering $I_*$ are the elements of
  $\operatorname{N}(I_i, I_*)$ and also just before considering $I_*$
  all the latest elements of $\operatorname{N}(\cdot, I_*)$ are in the greedy
  state. See \cref{fig:merge_lemma}.
\end{proof}

\FIGURE{t}{0.45}{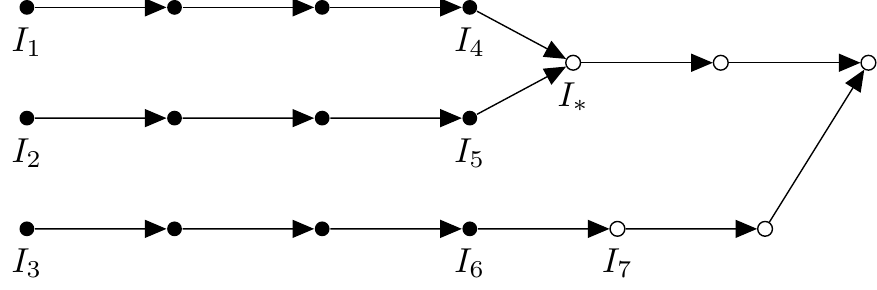}{
  \cref{lemma:greedy_tree_skip} for $G = \{I_1, I_2, I_3\}$.
  Here we assume $I_6 \prec I_* \prec I_7$.
  Elements of $\operatorname{N(\cdot, I_*)}$ are filled dots.
  We have $\NextG^9(G) = \{I_4, I_5, I_6\}$ and $I_* \in \NextG^{10}(G)$.
  \label{fig:merge_lemma}
}

If the elements of greedy state $G$ are all active, we can naively compute $I_*$
as in \cref{lemma:greedy_tree_skip} by checking LCAs of all pairs of
intervals in $G$ in the greedy tree and then proceeding to the last interval
before $I_*$ independently from each node to obtain the last greedy state before
reaching $I_*$ as in \cref{fig:merge_lemma}. Thus we have the following.
\begin{corollary}
  Let $G = \{I_1, I_2, \ldots, I_m\}$ be a greedy state with only active
  intervals and let $I_*$ be defined as in
  \cref{lemma:greedy_tree_skip}. It is possible to compute both the smallest $k$
  for which $I_* \in \NextG^k(G)$ and the value of $\NextG^k(G)$ itself in
  $\Ohtilde(m^2)$ time.
  \label{prop:greedy_tree_skip}
\end{corollary}

\subsection{An $\Ohtilde(n^{1/2})$-time algorithm for two machines}
\label{subsection:two_machines_algo}

In this section, we focus on describing an efficient algorithm for dynamic
interval scheduling on two machines and prove the following.
\begin{theorem}
  There is a data structure for DIS2 that supports any sequence of $n$
  insert/delete/query on intervals in $\Ohtilde(n^{1/2})$ amortized time
  per each operation.
  \label{theorem:n12_two}
\end{theorem}

\begin{lemma}
  There are only three possible forms of a greedy state for two machines.
  \begin{itemize}
    \item[(a)] $\{I_1, I_2\}$ where $I_1 \prec I_2$, $I_1$ and $I_2$ are
    compatible and $I_2$ is active,
    \item[(b)] $\{I_1, I_2\}$ where $I_1$ ends earlier than $I_2$,
    $I_1$ and $I_2$ are overlapping and $I_2$ is active,
    \item[(c)] $\{I_1, I_2\}$ where an active interval $I_1$ is fully contained
    inside (an inactive) $I_2$.
  \end{itemize}
  \label{lemma:two-machine-cases}
\end{lemma}

\begin{proof}
  First we assume, without losing generality, that the greedy algorithm
  considered at least two intervals and resumes from the greedy state $G_2$ that
  is of the form (a) and $I_1 = (s_1, f_1)$, $I_2 = (s_2, f_2)$. One can easily
  prepend any instance of DIS+ with few intervals to achieve this.

  Let assume that the next accepted interval by
  the greedy algorithm is $I_3 = (s_3, f_3)$ and the next greedy state after
  $G_2$ is $G_3$.
  There are three cases (as on \cref{fig:two_machine_cases}):
  \begin{itemize}
    \setlength{\topsep}{0pt}
    \setlength{\itemsep}{0pt}
    \item[(aa)] $f_2 < s_3$ -- then $I_3$ is an active interval and
    $G_3 = \{I_1, I_3\}$ is of the form (a),
    \item[(ab)] $s_2 < s_3 < f_2$ -- then $I_3$ is an active interval and
    $G_3 = \{I_2, I_3\}$ is of the form (b),
    \item[(ac)] $s_3 < s_2$ -- then $I_3$ is an inactive interval and
    $G_3 = \{I_2, I_3\}$ is of the form (c).
  \end{itemize}
  We now proceed to similar analysis of what are the forms of next greedy states
  that can be reached from states of the form (b) and (c).

  If $G_2$ is of the form (b) then $I_3$ is either compatible with $I_2$ (case
  (ba)) and $G_3$ is of the form (a), or it overlaps with $I_2$ (case (bb))
  and $G_3$ is of the form (b).
  Note that $I_3$ can not overlap with $I_1$ as then $I_3$ would be rejected.

  Similarly, if $G_2$ is of the form (c) then $I_3$ is either compatible with
  $I_2$ (case (ca)) and $G_3$ is of the form (a), or it overlaps with $I_2$
  (case (cb)) and $G_3$ is of the form (b).

  No other forms than (a), (b) or (c) are reachable from (a) and this concludes
  the proof.
\end{proof}

\FIGURE{t}{0.9}{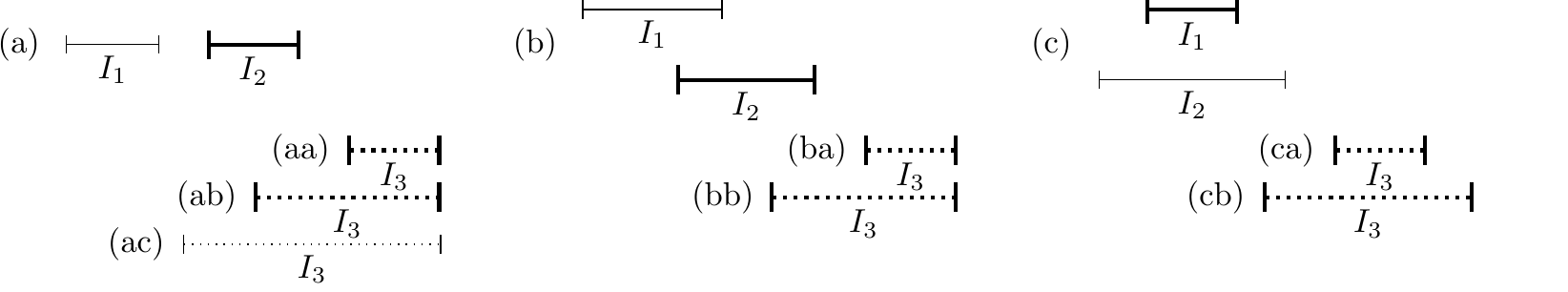}{
  Three possible forms of a greedy state and cases as in
  \cref{lemma:two-machine-cases}.
  Active intervals are marked with bold lines and potential cases for $I_3$ are
  marked with dotted lines.
  Note that $I_1$ in forms (a) and (b) may be either active or inactive.
  \label{fig:two_machine_cases}
}

\noindent
We now describe our algorithm for DIS2.
For each part it maintains the following:
\begin{itemize}
  \setlength{\topsep}{0pt}
  \setlength{\itemsep}{0pt}
  \item $B[I_i]$ for all active intervals $I_i$ -- the result of part query
  from the greedy state $\{I_i, I_{i'}\}$ of the form (b) where $I_{i'}$ is
  direct successor (in $\prec$ order),
  \item $C[I_i]$ for all inactive intervals $I_i$ -- the result of part query
  from the greedy state $\{I_i, I_{i'}\}$ of the form (c) where $I_{i'}$ is the
  latest (in $\prec$ order) active interval fully inside $I_i$.
\end{itemize}
When a part is updated, $B[\cdot]$ and $C[\cdot]$ structures are rebuilt
from scratch. Computation of $B[I_i]$ or $C[I_i]$ is nothing else than answering
a part query for the appropriate greedy state. We ask these queries in
decreasing order of the sum of indices (in $\prec$ order) of the two intervals
of the greedy state.
This way, during the recomputation of $B[\cdot]$ and $C[\cdot]$ structures,
whenever the algorithm is going to use some other result of $B[\cdot]$ or
$C[\cdot]$ it is already computed as the queried sum of indices will be
larger. See the details below.

Additionally, for each active interval $I_i$ we precompute the earliest
(in $\prec$) interval $I_{i'}$ not on the path from $I_i$ to the root of
$\TS_j$.
We do this using dynamic programming, inspecting all the intervals
in decreasing order of $\prec$ and it takes $\Ohtilde(n^{1/2})$ time.
Similarly, we precompute the number of intervals on the path from $I_i$ to the
latest interval ending earlier than $I_{i'}$.

We now describe how to answer the part query from a greedy state $G_i$ following
the proof of \cref{lemma:two-machine-cases} and considering all
forms of $G_i$.

If $G_i$ is of the form (a) we focus on finding the greedy state
$G_{i'} = \NextG^k(G_i)$ for which $k$ is the smallest such that
$I_1 \not \in \NextG^k(G_i)$. If $I_1$ is replaced in $G_{i'}$ by an active
interval $I_{i'}$, it has to be the earliest (in $\prec$) interval overlapping
with an interval $I_{i''}$ on the path from $I_2$ to the root in $\TS_j$
(it can also be $I_2$ itself).
We know which one and what is the contribution to the internal result as we
precomputed it.
Moreover, we observe that $G_{i'} = \{I_{i'}, I_{i''}\}$ and its part result is
stored in $B[I_{i''}]$ so we just read the result from there.
If $I_1$ is replaced in $G_{i'}$ by an inactive interval $I_{i'}$ it has to be
the earliest (in $\prec$) interval compatible with $I_1$. Then
$G_{i'} = \{I_{i'}, I_{i''}\}$ where $I_{i''}$ is the latest (in $\prec$ order)
active interval fully inside $I_{i'}$. Thus, we read the part result for $G_{i'}$
from $C[I_{i'}]$.

If $G_i$ is of the form (b), then $\NextG(G_i)$ is either of the form (a) for
which we proceed as described above or of the form (b) but with both greedy
state intervals active (case (bb) of the proof of
\cref{lemma:two-machine-cases}), for which we use \cref{prop:greedy_tree_skip}
to reach the greedy state of the form (a) and later proceed as described above.

If $G_i$ is of the form (c), then $\NextG(G_i)$ is either of the form (a) or (b)
and we proceed as described above.

\subsection{An $\Ohtilde(n^{1-1/m})$-time algorithm for $m \ge 3$ machines}

Surprisingly, before we start describing the final algorithm for $m \ge 3$
machines, we need an additional building block for the two machine case.

\begin{definition}
  For a collection of intervals $S$, for $I_1 \prec I_2$ from $S$, we define
  the \emph{first machine replacement} $\M(I_1, I_2)$ to be the interval in $S$
  which replaces $I_1$ in the greedy state when resumed the greedy execution
  from the greedy state $\{I_1, I_2\}$ on two machines. In other words,
  $\M(I_1, I_2)$ is the earliest ending accepted interval after $I_2$ that will
  be scheduled on the same machine as $I_1$ by the greedy algorithm for IS+.
\end{definition}

\FIGURE{t}{0.35}{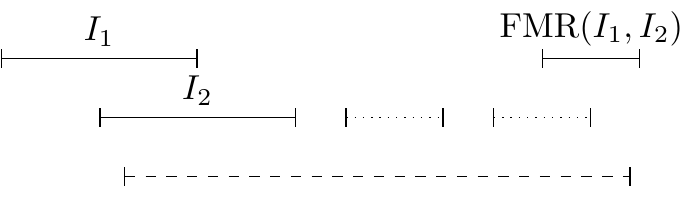}{Both dotted intervals are accepted by the machine
that accepted $I_2$ and the dashed interval overlaps with $I_1$ so is
rejected. The left dotted interval in the example is $I_{i'''}$, the
solution of the subproblem from the computation of $\M(I_1, I_2)$.}

Within the desired time bounds, for $m \ge 3$, we can afford recomputing
$\M(\cdot, \cdot)$
in parts from scratch for every pair of intervals in the updated part, as long
as this recomputation takes $\Ohtilde(|\PS_j|^2)$ time.
We could not do the same for $m = 2$.

\begin{lemma}
  The values of $\M(\cdot, \cdot)$ for all pairs of intervals in a collection of
  $n$ intervals, can be computed in $\Ohtilde(n^2)$ time.
\end{lemma}
\begin{proof}
  Assuming that intervals in part are ordered by $\prec$ and given names
  $I_1, I_2, \ldots$ in line with this order, we compute $\M(I_{i'}, I_{i''})$
  in decreasing order of the sum of $i' + i''$ indices.
  To compute $\M(I_{i'}, I_{i''})$, for $I_{i'} = (s_{i'}, f_{i'})$ and
  $I_{i''} = (s_{i''}, f_{i''})$ we first find the earliest ending interval
  $I_{i'''}$ that ends later than $I_{i''}$ and is compatible with $I_{i'}$.
  To solve this subproblem we take a geometric view: each interval
  $(s_i, f_i)$ is converted into a point $(s_i, f_i)$ in 2D plane, the goal is
  to find the point with smallest $y$-coordinate above and to the
  right of $(s_{i'}, f_{i''})$. This is solved by a 2D range search
  tree indexed by $(x, y)$-coordinates storing the appropriate result.
  Thus, the subproblem is solved.
  We proceed with the computation of $\M(I_{i'}, I_{i''})$.
  We have two cases: either $I_{i'''}$ is overlapping with $I_{i''}$
  and then $\M(I_{i'}, I_{i''}) = I_{i'''}$ or $I_{i'''}$ is compatible with
  $I_{i''}$ and then $\M(I_{i'}, I_{i''}) = \M(I_{i'}, I_{i'''})$ which is
  already known by the order of the computation.
  We can also compute the number of intervals chosen by the greedy algorithm
  when resumed from state $\{I_{i'}, I_{i''}\}$ until reaching
  $\M(I_{i'}, I_{i''})$ (just $1$ or the number chosen from $I_{i'}, I_{i'''}$
  plus $1$ depending on the above cases).
\end{proof}

As it turns out, the values of $\M(\cdot, \cdot)$ play important role in the
algorithm for $m \ge 3$.
We want to preprocess tuples of possible entry greedy states for
a part to be able to efficiently answer part queries. The problem is that we
have $\Oh(n^{1/m})$ intervals in each part, but we aim at $\Ohtilde(n^{1-1/m})$
time complexity. Thus, we cannot precompute part queries for
all possible greedy states. Instead, we carefully select specific
\emph{compressible} greedy states for which part query results are actually
stored and design an algorithm that can push the simulation forward to the next
compressible state or the exit state from the part.

\begin{definition}
  Let $G_t = \{I_1, I_2, \dots, I_m\}$ be a greedy state. We assume
  $I_1 \preceq I_2 \preceq \ldots \preceq I_m$.
  We say that $G_t$ is \emph{compressible} if at least one of the following
  conditions hold:
  \begin{itemize}
    \setlength{\itemsep}{0pt}
    \item[(a)] $I_m$ is inactive,
    \item[(b)] $I_m$ is active and exists active interval $I_p$ such that
    $\LC(I_p) = I_m$,
    \item[(c)] $I_m = \M(I_1, I_{m-1})$.
  \end{itemize}
  \label{def:compressible}
\end{definition}

\begin{lemma}
  In a~part of $n$ intervals for DIS+ on $m$ machines, there are only
  $\Oh(n^{m-1})$ compressible greedy states.
  \label{lemma:compressible}
\end{lemma}
\begin{proof}
  We consider all the forms of the compressible greedy state as in
  \cref{def:compressible}.
  \begin{itemize}
    \setlength{\itemsep}{0pt}
    \item[(a)] from \cref{lemma:inactive_active_inside} we know that the
    greedy state also contains the latest interval fully inside $I_m$ and thus
    we can forget this interval, so there are $\Oh(n^{m-1})$ such states,
    \item[(b)] there is an edge $(I_p, I_m)$ in the greedy tree of $\PS_j$, we
    can store $(m-2)$-tuple of other intervals and the identifier of the
    appropriate edge, so there are $\Oh(n^{m-1})$ such states,
    \item[(c)] we forget $I_m$ as it is equal to $\M(I_1, I_{m-1})$,
    so there are $\Oh(n^{m-1})$ such states. \qedhere
  \end{itemize}
\end{proof}

Note that we can decompress the representations from \cref{lemma:compressible}
in $\Oh(m)$ time to obtain a full greedy state of size $m$. Also, by taking into
account the sizes of the parts, we obtain that there are only $\Oh(n^{1-1/m})$
compressible greedy states for $n$ intervals in $S$.

For an update in $\PS_j$, we recompute part query results for all
compressible greedy states in $\PS_j$. As in
\cref{subsection:two_machines_algo}, we do this using dynamic programming, in
decreasing order of the sum of indices of the uncompressed state.
The problem of computing the results for the states stored in the dynamic
programming table is once again translated into the general query that has
$m$-tuple as an input and has to push the simulation forward either to the next
part or at least to a compressible greedy state from which we read the already
preprocessed result and combine it with the traversed prefix of the path.
We proceed with describing how to solve this general query.

We distinguish three forms of the greedy state $G = \{I_1, I_2, \dots, I_m\}$
for $I_1 \preceq I_2 \preceq \ldots \preceq I_m$:
\begin{itemize}
  \setlength{\itemsep}{0pt}
  \item[(*)] $I_m$ is inactive,
  \item[(**)] there is $1 < p \le m$ such that all intervals
  $I_p, I_{p+1}, \dots, I_m$ are active,
  \item[(***)] all $I_1, I_2, \dots, I_m$ are active.
\end{itemize}

For the (*) case, we compute $G' = \NextG(G)$.
Either the latest accepted interval in $G'$ is inactive and then $G'$ is
compressible of type (a) or it is active, thus $G'$ is either of the (**) or
(***) form and we proceed with it as described below.

For the (***) case, we use \cref{prop:greedy_tree_skip} to find the
earliest greedy state $G' = \NextG^k(G)$ for which $I_* \in G'$.
We observe that such $G'$ is compressible of type (b), as both $I_*$ and at
least one of its children are elements of $G'$.

We now consider the (**) case. We compute $G' = \NextG(G)$ and consider the
following subcases depending on the latest accepted interval $I_+$ in $G'$:
\begin{itemize}
  \setlength{\itemsep}{0pt}
  \item[(1)] $I_+$ is inactive,
  \item[(2)] $I_+$ is active, overlapping with $I_m$ and $I_1 \in G'$,
  \item[(3)] $I_+$ is active, overlapping with $I_m$ and $I_1 \not \in G'$,
  \item[(4)] $I_+$ is active and not overlapping with $I_m$.
\end{itemize}

In case (1), we see that $G'$ is compressible of type (a).
In case (2), we observe that $I_+ = \M(I_1, I_m)$, so $G'$ is compressible of
type (c).
In case (3), we observe that $G'$ remains of type (**), but with smaller $p$.
We proceed with computing $\NextG(G')$ until we reach any other case, which
happens after at most $\Oh(m)$ iterations.
In case (4), we observe that $G' = G \setminus \{I_m\} \cup \{I_+\}$ and
$I_+$ is compatible with every other interval from $G'$.
We read $\M(I_1, I_+)$ and push the simulation forward until reaching the first
greedy state $G''$ with accepted interval $I_{++}$ that will be scheduled on
a different machine than $I_m$. Notice that if $I_{++}$ is active then it is
compatible with $I_1$ so $I_{++} = \M(I_1, I_m)$. As $m \ge 3$, $I_{++}$ will
replace $I_{m-1} \neq I_1$ in the greedy state thus $G''$ is compressible of
type (c) and if $I_{++}$ is inactive then $G''$ is compressible of type (a).

\bibliography{references}

\begin{thebibliography}{10}

\bibitem{alstrup}
Stephen Alstrup, Jacob Holm, Kristian~De Lichtenberg, and Mikkel Thorup.
\newblock Maintaining information in fully dynamic trees with top trees.
\newblock {\em ACM Trans. Algorithms}, 1(2):243–264, October 2005.
\newblock \href {https://doi.org/10.1145/1103963.1103966}
  {\path{doi:10.1145/1103963.1103966}}.

\bibitem{amir}
Amihood Amir and Itai Boneh.
\newblock Dynamic suffix array with sub-linear update time and poly-logarithmic
  lookup time.
\newblock {\em CoRR}, abs/2112.12678, 2021.

\bibitem{arkin}
Esther~M. Arkin and Ellen~B. Silverberg.
\newblock Scheduling jobs with fixed start and end times.
\newblock {\em Discrete Applied Mathematics}, 18(1):1 -- 8, 1987.
\newblock URL:
  \url{http://www.sciencedirect.com/science/article/pii/0166218X87900370},
  \href {https://doi.org/https://doi.org/10.1016/0166-218X(87)90037-0}
  {\path{doi:https://doi.org/10.1016/0166-218X(87)90037-0}}.

\bibitem{bender}
Michael~A Bender, Richard Cole, Erik~D Demaine, Martin Farach-Colton, and Jack
  Zito.
\newblock Two simplified algorithms for maintaining order in a list.
\newblock In {\em European Symposium on Algorithms}, pages 152--164. Springer,
  2002.

\bibitem{bouzina}
Khalid~I. Bouzina and Hamilton Emmons.
\newblock Interval scheduling on identical machines.
\newblock {\em Journal of Global Optimization}, 9(3):379--393, Dec 1996.
\newblock \href {https://doi.org/10.1007/BF00121680}
  {\path{doi:10.1007/BF00121680}}.

\bibitem{carlisle}
Martin~C. Carlisle and Errol~L. Lloyd.
\newblock On the k-coloring of intervals.
\newblock {\em Discrete Applied Mathematics}, 59(3):225 -- 235, 1995.
\newblock URL:
  \url{http://www.sciencedirect.com/science/article/pii/0166218X9580003M},
  \href {https://doi.org/https://doi.org/10.1016/0166-218X(95)80003-M}
  {\path{doi:https://doi.org/10.1016/0166-218X(95)80003-M}}.

\bibitem{Compton}
Spencer Compton, Slobodan Mitrović, and Ronitt Rubinfeld.
\newblock New partitioning techniques and faster algorithms for approximate
  interval scheduling, 2020.
\newblock \href {http://arxiv.org/abs/2012.15002} {\path{arXiv:2012.15002}}.

\bibitem{cormen}
Thomas~H. Cormen, Charles~E. Leiserson, Ronald~L. Rivest, and Clifford Stein.
\newblock {\em Introduction to Algorithms, Third Edition}.
\newblock The MIT Press, 3rd edition, 2009.

\bibitem{dietz}
P.~Dietz and D.~Sleator.
\newblock Two algorithms for maintaining order in a list.
\newblock In {\em Proceedings of the Nineteenth Annual ACM Symposium on Theory
  of Computing}, STOC '87, page 365–372, New York, NY, USA, 1987. Association
  for Computing Machinery.
\newblock \href {https://doi.org/10.1145/28395.28434}
  {\path{doi:10.1145/28395.28434}}.

\bibitem{doll}
Christof Doll, Tanja Hartmann, and Dorothea Wagner.
\newblock Fully-dynamic hierarchical graph clustering using cut trees.
\newblock In Frank Dehne, John Iacono, and J{\"o}rg-R{\"u}diger Sack, editors,
  {\em Algorithms and Data Structures}, pages 338--349, Berlin, Heidelberg,
  2011. Springer Berlin Heidelberg.

\bibitem{faigle}
Ulrich Faigle and Willem~M. Nawijn.
\newblock Note on scheduling intervals on-line.
\newblock {\em Discrete Applied Mathematics}, 58(1):13 -- 17, 1995.
\newblock URL:
  \url{http://www.sciencedirect.com/science/article/pii/0166218X95001125},
  \href {https://doi.org/https://doi.org/10.1016/0166-218X(95)00112-5}
  {\path{doi:https://doi.org/10.1016/0166-218X(95)00112-5}}.

\bibitem{gavruskin}
Alexander Gavruskin, Bakhadyr Khoussainov, Mikhail Kokho, and Jiamou Liu.
\newblock Dynamic algorithms for monotonic interval scheduling problem.
\newblock {\em Theoretical Computer Science}, 562:227 -- 242, 2015.
\newblock URL:
  \url{http://www.sciencedirect.com/science/article/pii/S0304397514007361},
  \href {https://doi.org/https://doi.org/10.1016/j.tcs.2014.09.046}
  {\path{doi:https://doi.org/10.1016/j.tcs.2014.09.046}}.

\bibitem{gawrychowskijanczewski}
Pawel Gawrychowski and Wojciech Janczewski.
\newblock Fully dynamic approximation of {LIS} in polylogarithmic time.
\newblock In Samir Khuller and Virginia~Vassilevska Williams, editors, {\em
  {STOC} '21: 53rd Annual {ACM} {SIGACT} Symposium on Theory of Computing,
  Virtual Event, Italy, June 21-25, 2021}, pages 654--667. {ACM}, 2021.
\newblock \href {https://doi.org/10.1145/3406325.3451137}
  {\path{doi:10.1145/3406325.3451137}}.

\bibitem{HenzingerKNS15}
Monika Henzinger, Sebastian Krinninger, Danupon Nanongkai, and Thatchaphol
  Saranurak.
\newblock Unifying and strengthening hardness for dynamic problems via the
  online matrix-vector multiplication conjecture.
\newblock In {\em {STOC}}, pages 21--30. {ACM}, 2015.

\bibitem{henzinger}
Monika Henzinger, Stefan Neumann, and Andreas Wiese.
\newblock {Dynamic Approximate Maximum Independent Set of Intervals, Hypercubes
  and Hyperrectangles}.
\newblock In Sergio Cabello and Danny~Z. Chen, editors, {\em 36th International
  Symposium on Computational Geometry (SoCG 2020)}, volume 164 of {\em Leibniz
  International Proceedings in Informatics (LIPIcs)}, pages 51:1--51:14,
  Dagstuhl, Germany, 2020. Schloss Dagstuhl--Leibniz-Zentrum f{\"u}r
  Informatik.
\newblock URL: \url{https://drops.dagstuhl.de/opus/volltexte/2020/12209}, \href
  {https://doi.org/10.4230/LIPIcs.SoCG.2020.51}
  {\path{doi:10.4230/LIPIcs.SoCG.2020.51}}.

\bibitem{henzinger2}
Monika~R. Henzinger and Valerie King.
\newblock Randomized fully dynamic graph algorithms with polylogarithmic time
  per operation.
\newblock {\em J. ACM}, 46(4):502–516, jul 1999.
\newblock \href {https://doi.org/10.1145/320211.320215}
  {\path{doi:10.1145/320211.320215}}.

\bibitem{holm}
Jacob Holm, Kristian de~Lichtenberg, and Mikkel Thorup.
\newblock Poly-logarithmic deterministic fully-dynamic algorithms for
  connectivity, minimum spanning tree, 2-edge, and biconnectivity.
\newblock {\em J. ACM}, 48(4):723–760, jul 2001.
\newblock \href {https://doi.org/10.1145/502090.502095}
  {\path{doi:10.1145/502090.502095}}.

\bibitem{kempakociumaka}
Dominik Kempa and Tomasz Kociumaka.
\newblock Dynamic suffix array with polylogarithmic queries and updates.
\newblock {\em CoRR}, abs/2201.01285, 2022.
\newblock URL: \url{https://arxiv.org/abs/2201.01285}, \href
  {http://arxiv.org/abs/2201.01285} {\path{arXiv:2201.01285}}.

\bibitem{kleinbergtardos}
Jon Kleinberg and Eva Tardos.
\newblock {\em Algorithm Design}.
\newblock 01 2006.

\bibitem{kociumakaseddighin}
Tomasz Kociumaka and Saeed Seddighin.
\newblock Improved dynamic algorithms for longest increasing subsequence.
\newblock In {\em Proceedings of the 53rd Annual ACM SIGACT Symposium on Theory
  of Computing}, STOC 2021, page 640–653, New York, NY, USA, 2021.
  Association for Computing Machinery.
\newblock \href {https://doi.org/10.1145/3406325.3451026}
  {\path{doi:10.1145/3406325.3451026}}.

\bibitem{kolen}
Antoon~W.J. Kolen, Jan~Karel Lenstra, Christos~H. Papadimitriou, and Frits~C.R.
  Spieksma.
\newblock Interval scheduling: A survey.
\newblock {\em Naval Research Logistics (NRL)}, 54(5):530--543, 2007.
\newblock URL: \url{https://onlinelibrary.wiley.com/doi/abs/10.1002/nav.20231},
  \href
  {http://arxiv.org/abs/https://onlinelibrary.wiley.com/doi/pdf/10.1002/nav.20231}
  {\path{arXiv:https://onlinelibrary.wiley.com/doi/pdf/10.1002/nav.20231}},
  \href {https://doi.org/https://doi.org/10.1002/nav.20231}
  {\path{doi:https://doi.org/10.1002/nav.20231}}.

\bibitem{lincolnwilliams2}
Andrea Lincoln, Virginia~Vassilevska Williams, and Ryan Williams.
\newblock {\em Tight Hardness for Shortest Cycles and Paths in Sparse Graphs},
  pages 1236--1252.
\newblock URL: \url{https://epubs.siam.org/doi/abs/10.1137/1.9781611975031.80},
  \href
  {http://arxiv.org/abs/https://epubs.siam.org/doi/pdf/10.1137/1.9781611975031.80}
  {\path{arXiv:https://epubs.siam.org/doi/pdf/10.1137/1.9781611975031.80}},
  \href {https://doi.org/10.1137/1.9781611975031.80}
  {\path{doi:10.1137/1.9781611975031.80}}.

\bibitem{demainepatrascu}
Mihai Patrascu and Erik~D. Demaine.
\newblock Logarithmic lower bounds in the cell-probe model.
\newblock {\em SIAM Journal on Computing}, 35(4):932--963, 2006.
\newblock \href
  {http://arxiv.org/abs/https://doi.org/10.1137/S0097539705447256}
  {\path{arXiv:https://doi.org/10.1137/S0097539705447256}}, \href
  {https://doi.org/10.1137/S0097539705447256}
  {\path{doi:10.1137/S0097539705447256}}.

\bibitem{willard}
Dan~E. Willard.
\newblock Log-logarithmic worst-case range queries are possible in space
  {$\Theta(N)$}.
\newblock {\em Information Processing Letters}, 17(2):81 -- 84, 1983.
\newblock URL:
  \url{http://www.sciencedirect.com/science/article/pii/0020019083900753},
  \href {https://doi.org/https://doi.org/10.1016/0020-0190(83)90075-3}
  {\path{doi:https://doi.org/10.1016/0020-0190(83)90075-3}}.

\end{thebibliography}

\appendix

\section{An $\Ohtilde(n^{1/2})$-time algorithm for one machine}
\label{section:n12_solution}

Here we show a simple structure that already needs a subset of the ideas used in
more complicated and faster structures described in the paper.
We show the data structure for DIS showing the following.
\begin{theorem}
  There is a data structure for DIS that supports any sequence of $n$
  insert/delete operations in $\Oh(\sqrt{n} \log n)$ amortized time per
  update.
  \label{thm:sqrtlog}
\end{theorem}

We choose the set of \emph{separators} $x_1 < x_2 < \dots < x_k$. Separators
split the universe of coordinates into $k+1$ \emph{parts}:
$\PS_0, \PS_1, \dots, \PS_k$.
Assuming that $x_{k+1} = +\infty$, each part $\PS_j$ contains intervals
$(s_i, f_i)$ with $x_j \le s_i < x_{j+1}$.
Thus, at any time parts represent partition of intervals.
We define the \emph{size} of a part as the number of intervals in it.
Intervals store references to parts in which are contained.

Separators (and pointers to the appropriate parts) are stored in
a predecessor/successor data structure (we use balanced binary search trees
\cite{cormen}
\footnote{If for all intervals $s_i, f_i$ are small integers bounded by $U$, we
could use $y$-fast tries \cite{willard}. This way we could achieve
$\Oh(\sqrt{n} \log \log U)$ amortized time per operation.})
and are chosen to satisfy the following invariant:
each part has size at most $2\sqrt{N}$ and for every two
consecutive parts $\PS_j$ and $\PS_{j+1}$ at least one has size at least
$\frac{1}{2}\sqrt{N}$. Thus, there are $\Oh(\sqrt{n})$ parts at any time
and local rebuild of parts of size $\Oh(\sqrt{n})$ happens after
$\Omega(\sqrt{n})$ operations affecting the part. As these rebuilds are
simply appropriate separate insertions, the amortized update time complexity
does not change.

Intervals in part $\PS_j$ satisfying $f_i < x_{j+1}$ are called \emph{internal}
and all the others are called \emph{external}.
Internal intervals are stored in predecessor/successor data structures: sorted
by $s_i$ and, separately, sorted by $f_i$.
Additionally, we have the same structures defined globally, for all the
intervals in $S$. This allows to compute $\LC(I_i)$ in $\Oh(\log n)$ time.

For each internal interval $I_i$ in part $\PS_j$ we store its leftmost
compatible internal interval in the same part, denoted by $\LCinternal(I_i)$
(either $\LC(I_i)$ or $\perp$ in case $\LC(I_i) \not \in \PS_j$ or is external).
Additionally, we store the information to resume the greedy execution from an
internal interval $I_i$ to the latest interval in the same part. This includes:
$\Rinternal(I_i)$ -- the largest $r \ge 0$ such that
$\LCinternal^r(I_i) \neq \perp$ and
$\Einternal(I_i) = \LC^{\Rinternal(I_i)-1}(I_i)$.

When the content of $\PS_j$ is updated, all the above values for intervals of
$\PS_j$ are recomputed naively from scratch: we start with computing
$\LCinternal(\cdot)$ in decreasing order of $f_i$.
We set $\LCinternal((s_i, f_i))$ to be the interval $I'$ with the smallest
$f_{i'}$ among intervals with $s_{i'} \ge f_i$ or $\perp$ if there is no such
interval.
We update which interval is $I'$ whenever the computation of
$\LCinternal(\cdot)$ proceeds to smaller values of $f_i$ by querying the
appropriate part structure (containing only internal intervals) sorted by $s_i$.
Overall, this naive recomputation of all information for all internal intervals
in the part takes $\Oh(\sqrt{n} \log n)$ time.

With the above, we can resume the greedy algorithm from any $I_i$ in
any $\PS_j$ until reaching the earliest interval in the solution outside $\PS_j$
in $\Oh(\log n)$ time.
For an external interval $I_i$ it is enough to proceed to $\LC(I_i)$ to exit
$\PS_j$.
If $I_i$ is internal, we increase the total result by the number of selected
intervals in the part (the internal result for $I_i$) and proceed to
$I_{i'} = \LC(\Einternal(I_i))$.
$I_{i'}$ may already be in some further part or it may be an external interval
in $\PS_j$ and then we proceed to $\LC(I_{i'}) \not \in \PS_j$.

To answer a DIS query, we simulate the execution of the greedy
algorithm starting from the earliest ending interval and traversing the parts as
described above. The query as described takes $\Oh(\sqrt{n} \log n)$ time.

To insert an interval $(s_i, f_i)$ to $S$, we first locate the appropriate part
$\PS_j$ in the separators structure, insert the interval into $\PS_j$, recompute
the additional information associated with $\PS_j$ and update the global
structures.
All these takes $\Oh(\sqrt{n} \log n)$ time.
During insertion, it may happen that $\PS_j$ becomes too large. In this case,
if the size reached $s$, we naively find (using an appropriate
predecessor/successor structure) $\lceil \frac{s}{2} \rceil$-th value $x$ in
the set of $s_i$s of all intervals in $\PS_j$ and add $x$ as the new separator.
This splits $\PS_j$ into two new parts, which we recompute from scratch.
This, again, works in $\Oh(\sqrt{n} \log n)$ time.

Deletion of an interval of $\PS_j$ is similar and in the case of underflow
of pair $\PS_j$ and $\PS_{j+1}$ or $\PS_{j-1}$ and $\PS_j$, we merge the parts
by removing the separator between them and recompute the new part.

\section{Lower bound for Dynamic Weighted Interval Scheduling}

The \textsc{Minimum Weight $k$-Clique} problem is to find, in an edge-weighted
graph, a clique of exactly $k$ nodes having the minimum total weight of edges.

The following hypothesis about \textsc{Minimum Weight $k$-Clique} problem was
formulated.
\begin{conjecture}[Min Weight $(2\ell+1)$-Clique Hypothesis \cite{lincolnwilliams2}]
  There is a constant $c > 1$ such that, on a Word-RAM with $\Oh(\log n)$-bit
  words, finding a $k$-Clique of minimum total edge weight in an $n$-node graph
  with non-negative integer edge weights in $[1, n^{ck}]$ requires $n^{k-o(1)}$
  time.
\end{conjecture}

The \textsc{Minimum Weight $(2\ell + 1)$-Cycle} problem is to find, in an
edge-weighted graph, a cycle consisting exactly $2\ell + 1$ edges having the
minimum total weight.

\begin{theorem}[\cite{lincolnwilliams2}]
  If there is an integer $\ell \ge 1$ and a constant $\epsilon > 0$ such that
  \textsc{Minimum Weight $(2\ell + 1)$-Cycle} in a directed weighted $n$-node
  $m = \Theta(n^{1+1/\ell})$-edge graph can be solved in
  $O(mn^{1-\epsilon} + n^2)$ time, then the Min Weight $(2\ell+1)$-Clique
  Hypothesis is false.
  \label{thm:minweightkcyclehard}
\end{theorem}

Based on the above, we formulate the following.
\begin{theorem}
  Unless the Min Weight $(2\ell+1)$-Clique Hypothesis is false, for all
  $\epsilon > 0$ there is no algorithm for DWIS problem with $O(n^{1-\epsilon})$
  update and query time.
\end{theorem}
\begin{proof}
  As in \cite{lincolnwilliams2}, we use the fact that
  \textsc{Minimum Weight $(2\ell + 1)$-Cycle} is still hard if restricted only
  to $k$-circle layered graphs, that is $k$-partite graphs in which, for each
  $i \in [k]$, all edges from nodes in $i$-th part end in $(i \bmod k + 1)$-th
  part.

  We reduce \textsc{Minimum Weight $(2\ell + 1)$-Cycle} in a weighted
  $(2\ell + 1)$-circle layered graph to DWIS.
  The input instance has $n$ nodes and $m = \Theta(n^{1+1/\ell})$
  edges of integer weights in range $[n^{c\ell} = W]$ for large enough $c$.
  We enumerate parts from $1$ to $2\ell + 1$ and we enumerate nodes
  independently in each parts starting from $0$.
  For all $p \in [2\ell]$, for all edges from $u$-th node in $p$-th part
  to $v$-th node in $(p+1)$-th part,
  we insert an interval $[(p-1)n + u,\ pn + v)$ of weight
  $(f_i - s_i)(2\ell + 1)(W + 1) + (W - w)$ where $w$ is the edge weight.

  The optimal cycle has to go through some node $s$ in the first part.
  We guess this node by inserting an interval $[-1, s)$ of weight
  $(f_i - s_i)(2\ell + 1)(W + 1)$ and, for all edges from $u$-th
  node in $(2\ell + 1)$-th part to $s$ of weight $w$, we insert
  an interval $[2\ell \cdot n + u,\ (2\ell+1) \cdot n)$ of weight
  $(f_i - s_i)(2\ell + 1)(W + 1) + (W - w)$.
  To start with another choice of $s$, we delete the corresponding
  intervals before inserting the new ones.

  The selection of edge weights in our instance guarantees that the
  optimal solution maximizes the total length of chosen intervals and then
  minimizes the weight resulting from weights of edges in the graph, as each
  unit of length increases the value of the solution by $(2\ell + 1)(W + 1)$
  while the additional gain from edge weights is, in total, at most
  $(2\ell + 1)W$.

  The only possibility to obtain the value of at least
  $(2\ell + 1)(W + 1)n$ is to choose the intervals spanning the whole interval
  $[-1, (2\ell + 1)n)$ in the created instance. Such selection ensures that
  an interval representing node $s$ in the first part is selected, as well as
  all intervals representing the edges on the cycle, including the last edge
  going to the first part represented by the interval with
  $f_i = (2\ell + 1)(W + 1)$. Because in this scenario there is no gap nor
  overlap in coordinates of the selected intervals, any two consecutive edges
  share a common node, so they form a $(2\ell + 1)$-cycle.
  Thus, there is 1-1 correspondence between $(2\ell + 1)$-cycles going through
  node $s$ in the first part and solutions of weight at least
  $(2\ell + 1)(W + 1)n$.
  Nodes of the optimal $(2\ell + 1)$-cycle can be deduced by
  inspecting endpoints of the selected intervals.

  To solve \textsc{Minimum Weight $(2\ell + 1)$-Cycle} by the above reduction we
  invoked $\Oh(m)$ insertions and deletions to DWIS structure. By choosing the
  input instance to have $\ell = \frac{1}{\epsilon}$ and $n = c\ell$ for large
  enough $c$, and assuming (ad absurdum) that these
  $\Oh(m) = \Oh(n^{1+\epsilon})$ operations took $\Oh(m \cdot m^{1-\epsilon})$
  time, we obtained $\Oh(n^{2+\epsilon-\epsilon^2})$-time algorithm for the
  \textsc{Minimum Weight $(2\ell + 1)$-Cycle problem}, thus violating
  \cref{thm:minweightkcyclehard}.
\end{proof}

\end{document}